\documentclass[10pt]{amsart}
\usepackage{amsmath,amsthm,amsfonts}
\usepackage{graphicx}
\usepackage[usenames]{color}
\usepackage{hyperref}

\newenvironment{red}{\color{red}}{}
\newenvironment{blue}{\color{blue}}{}
\newenvironment{green}{\color{green}}{}
\newtheorem{theorem}{Theorem}[section]
\newtheorem{lemma}[theorem]{Lemma}
\theoremstyle{definition}

\newtheorem{corollary}[theorem]{Corollary}

\newtheorem{definition}[theorem]{Definition}
\newtheorem{example}[theorem]{Example}
\newtheorem{problem}[theorem]{Problem}
\newtheorem{remark}[theorem]{Remark}
\newtheorem{proposition}[theorem]{Proposition}

\newcommand{\gp}[1]{{\left\langle #1 \right\rangle}}

\newcommand{\ygc}[1]{{\red [YG comment]}}
\newcommand{\rb}[1]{{\left( #1 \right)}}

\newcommand{\cclass}[1]{\mathbf{#1}}

\def\ovn{{\overline n}}

\def\CC{{\mathcal C}}
\def\CS{{\mathcal S}}

\def\P{{\mathbf{P}}}
\def\NP{{\mathbf{NP}}}
\def\GP{{\mathbf{GP}}}
\def\SGP{{\mathbf{SGP}}}
\def\DistNP{{\mathbf{DistNP}}}
\DeclareMathOperator{\supp}{{supp}}

\def\MN{{\mathbb{N}}}

\title{Generic case completeness}

\author[A. Miasnikov]{Alexei Miasnikov}

\author[A. Ushakov]{Alexander Ushakov}
\address{Department of Mathematics, Stevens Institute of Technology, Hoboken, NJ, USA}
\email{amiasnik,aushakov@stevens.edu}
\thanks{The work was partially supported by NSF grant DMS-1318716.}

\begin{document}
\maketitle

\begin{abstract}
In this note we introduce a notion of a generically (strongly generically) $\NP$-complete
problem and show that the randomized bounded version of the halting
problem is strongly generically $\NP$-complete.


\end{abstract}

\tableofcontents

\section{Introduction}
We introduce and study  problems that are generically in $\NP$, i.e.,
decision problems that have partial  errorless nondeterministic
decision algorithms that solve the problem in polynomial time on ``most'' inputs. We define appropriate  reductions in this class and show that there are some complete problems there, called {\em strongly generically $\NP$-complete} problems. In particular, the randomized bounded version of the halting problem is one of them.

Rigorous formulation of notions of generic algorithms and generic complexity appeared first in group theory \cite{KMSS1,KMSS2} as a response to several  challenges that algorithmic algebra faced at that time. First, it was well understood that many hard, even undecidable, algorithmic problems in groups can be easily solved on most instances (see \cite{KMSS1,KMSS2,GMMU:2007,MSU_book:2011} for a thorough discussion).
Second, the study of  random objects and generic properties of objects has become the mainstream of geometric group theory, following the lead of  graph and number theory (see \cite{Gromov_hyperbolic, Gr, Gromov:2003,Ol,Arzhantseva:1998,Champetier:1995,BMR}). It turned out that  ``random'', ``typical'' objects have many nice properties that lead to simple and efficient algorithms. However a rigorous formalization of this approach was lagging behind. Algorithmic algebra  was still focusing mostly on the worst-case complexity with minor inroads into average case complexity.
Third, with the rapid development of algebraic cryptography the quest for natural algorithmic problems, which are hard on most inputs, became one of the main subjects in complexity theory (see discussion in \cite{MSU_book:2011}). It was realized that the average case complexity does not fit well here.
Indeed,  by definition, one cannot consider average case complexity of undecidable problems, which are in the majority in group theory; the proofs of average case results are usually difficult and technical \cite{Gurevich:1991,Wang:1995}, and, most importantly,
there are problems that are provably hard on average but easy on most inputs (see \cite{GMMU:2007, MSU_book:2011} for details). In fact, Gurevich showed in \cite{Gurevich:1991} that the average case complexity is not about ``most'' or ``typical''  instances, but that it grasps the  notion of ``trade-off'' between the time of computation on hard  inputs and how many of such hard instances are there.
Nowadays, generic algorithms  form an organic part of computational algebra and play an essential role in practical computations.

In a surprising twist generic algorithms and ideas  of generic complexity
were recently adopted in abstract computability (recursion theory).
There is interesting and active research there concerning absolutely
undecidable problems, generic Turing degrees,
coarse computability, etc., relating  generic computation with deep structural properties of Turing degrees
\cite{MRyb, Jockusch-Schupp:2012,Bienvenu_Day_Holzl:2013,
Igusa:2013,Downey-Jockusch-Schupp:2013,Downey-Jockusch-McNicholl-Schupp:2014}.

We decided to relativize these ideas to lower complexity classes.
Here we consider the class $\NP$.
Motivation to study generically hardest problems in the class $\NP$ comes
from several areas of mathematics and computer science.
First, as we have mentioned above, average case complexity, even when it is high, does not give information on the hardness of the problem at hand on the typical or generic inputs. Therefore, to study hardness of the problem on most inputs one needs to develop a theory of generically complete problems in the class $\NP$. This is interesting in its own right, especially when much of activity in modern mathematics focuses on generic properties of mathematical objects and how to deal with them.
On the other hand, in modern crypotography, there is a quest for cryptoprimitives which are computationally hard to break on most inputs.
It would be interesting to analyze  which $\NP$-problems are hard on most inputs, i.e., which of them are generically $\NP$-complete. Note, there are $\NP$-complete problems that are generically polynomial \cite{MSU_book:2011}.
All this requires a  robust theory of generic $\NP$-completeness.
As the first attempt to develop such a theory we study here the class of all generically $\NP$-problems, their reductions, and the complete problems in the class.
Most of the time, our exposition follows the seminal Gurevich's paper \cite{Gurevich:1991}
on average complexity.
We conclude with several open problems that seem to be important for the theory.

Here we briefly describe  the structure of the paper and mention the main results.
In Section \ref{se:Preliminaries}, we recall some notions and introduce notation from
the classical decision problems.
In Section \ref{se:generic}, we discuss distributional decision problems (when the set of instances of
the problem comes equipped with some measure), then define the generic complexity and
problems decidable generically (strongly generically) in polynomial time.
In Section \ref{se:reductions}, we define generic polynomial time reductions.
In Section \ref{se:DistHP}, we show that the distributional bounded halting problem for Turing machines
is strongly generically $\NP$-complete.
Notice that though generic Ptime randomized algorithms are usually much easy to come up with
(than say  Ptime on average algorithms), the reductions in the class of  generic $\NP$-problems
are still as  technical as reductions in the class of $\NP$-problems on average. In fact, the reductions
in both classes are similar. Essentially, these are reductions among general randomized problems
and the main technical, as well as theoretical,  difficulty  concerns the transfer of the measure
when reducing one randomized problem to another one. It seems this difficulty is intrinsic to reductions
in randomized computations and does not depend on whether we consider generic or average complexity.
In Section \ref{se:Problems} we discuss some open problems that seem to be
important for the development of the theory of generic $\NP$-completeness.

\section{Preliminaries}
\label{se:Preliminaries}

In this section we introduce notation to follow throughout the paper.

\subsection{Decision problems}
\label{se:prelim}

 Informally, a {\em decision problem} is an arbitrary yes-or-no question for  an (infinite) set of {\em inputs} (or {\em instances})  $I$, i.e., an unary predicate $P$  on $I$. The problem is termed {\em decidable} if  $P$ is computable, and the main classical question is whether a given problem is decidable or not. In complexity theory the predicate $P$ usually is given by its  true  set $L = \{x \in I \mid P(x) = 1\}$, so the decision problem appears as a pair $(I,L)$.   Furthermore, it is assumed usually that  every input $x \in I$ admits a finite description in some finite alphabet $\Sigma$ in such a way that given a word $w \in \Sigma^\ast$ one can effectively determine if $w \in I$ or not. This allows one, without loss of generality, to assume simply that $I =  \Sigma^\ast$. Some care is required when dealing with  distributional problems and we discuss this issue in due course.  From now on, unless said otherwise,  we assume that decision problems are pairs $D = (\Sigma^\ast, L)$, where $L \subseteq \Sigma^\ast$. In this case  $\Sigma$ is the {\em alphabet} of the problem $D$ and we denote it sometimes by $\Sigma_D$; $\Sigma_D^\ast$ is the set of {\em inputs}  or the  {\em domain} of $D$; the set $L$ is the {\em yes} or {\em positive}  part of $D$, denoted sometimes by $D^{yes}$ or $D^+$. In Section \ref{subsec:restrictions} we briefly consider problems of the type $(I,L)$, where $L \subseteq I \subseteq \Sigma^\ast$, not assuming that $I$ is a decidable subset  of $\Sigma^\ast$.

It is natural now to define the {\em size} of $x \in \Sigma^\ast$ to be
its word length $|x|$.
As usual, we define the sphere $\Sigma^n$ of radius $n \in \mathbb{N}$ as the set
of all strings (words) in $\Sigma^\ast$ of size $n$,  and $D_n = D \cap \Sigma^n$.
For a symbol $a \in \Sigma$ and
$n\in\mathbb{N}$ put  $a^n$ to be the string of $n$ symbols $a$.

We assume that alphabet $\Sigma$ comes equipped with a fixed linear ordering. This allows one to introduce a {\em shortlex}  ordering $<_{slex}$ on  the set $\Sigma^\ast$ as follows. We order, first, the words in $\Sigma^\ast$ with respect to their length (size), and if two words have the same length then we compare them in the (left) lexicographical ordering. The successor of a word $x \in \Sigma^\ast$ is denoted by  $x^+$.

\subsection{Deterministic and nondeterministic Turing machines}

In this section we recall the definition of a Turing machine
in order to establish terminology.

\begin{definition}
A one-tape {\em Turing machine} (TM) $M$ is a $5$-tuple
$\gp{Q,\Sigma,q_0,q_1,\delta}$ where:
\begin{itemize}
    \item
$Q = \{q_0,q_1,\ldots,q_m\}$ is a finite set of {\em states};
    \item
$\Sigma = \{a_1,\ldots,a_n\}$ is a finite set called the {\em tape alphabet} which contains at least $2$ symbols;
    \item
$q_0 \in Q$ is the {\em initial state};
    \item
$q_1 \in Q$ is the {\em final state};
    \item
$\delta \subset Q \times (\Sigma\cup\{\sqcup\}) \times Q \times \Sigma \times \{L,R\}$
is the {\em transition relation}.
\end{itemize}
Additionally, $M$ uses a {\em blank symbol} $\sqcup$ different from the
symbols $\Sigma$ to mark the parts of the infinite tape not in
use. This is the only symbol allowed to occur on the tape infinitely
often at any step during the computation.
\end{definition}

We say that a transition relation $\delta$ in the definition of a TM
is {\em deterministic} if for every pair $(q,a) \in Q \times (\Sigma \cup \{\sqcup\})$ there is
a unique five-tuple $(q,a,q',\gamma',d)$ in $\delta$, i.e.,
$\delta$ defines a function $\delta^\ast: Q \times (\Sigma \cup \{\sqcup\})\rightarrow Q \times \Sigma \times \{L,R\}$.
We say that a TM $M$ is {\em deterministic} if its transition relation is.
Otherwise we say that $M$ is a nondeterministic machine (NTM).

Each Turing machine has a {\em tape} with $(\Sigma \cup \{\sqcup\})$-symbols written on it, a {\em head}
specifying a position on the tape, and a {\em state register} containing an element $q\in Q$.
We say that the head observes a symbol $a \in \Sigma \cup \{\sqcup\}$, if $a$ is written on the tape
at the position specified by the head. If a TM $M$ is in the state $q$ and observes a symbol
$a\in \Sigma$, then to perform a step of computations:
\begin{itemize}
    \item
$M$ chooses any element $(q,a,q',a',d)\in \delta$;
    \item
puts $q'$ into the state register;
    \item
writes $a'$ on the tape to the head position;
    \item
moves the head to left or to the right depending on $d$.
\end{itemize}
If $\delta$ contains no tuple $(q,a,q',a',d)$, then we say that $M$ {\em breaks}.

We can define the operation of a TM formally using the notion of a
configuration that contains a complete description of the current
state of computation. A {\em configuration} of $M$ is a triple
$(q,w,u)$ where $w,u$ are $\Sigma$-strings and $q \in Q$.
\begin{itemize}
    \item
$w$ is a string to the left of the head;
    \item
$u$ is the string to the right of the head, including the symbol
scanned by the head;
    \item
$q$ is the current state.
\end{itemize}

We say that a configuration $(q,w,u)$ {\em yields} a configuration
$(q',w',u')$ in one step, denoted by
    $$(q,w,u) \stackrel{M}{\rightarrow} (q',w',u'),$$
if a step of a machine from
configuration $(q,w,u)$ results in configuration $(q',w',u')$.
Note that if the machine is nondeterministic, then a configuration
can yield more than one configuration.
Using the relation ``yields in one step'' one can define relations
``yields in $k$ steps'', denoted by
    $$(q,w,u) \stackrel{M^k}{\rightarrow} (q',w',u'),$$
and ``yields'', denoted by
    $$(q,w,u) \stackrel{M^\ast}{\rightarrow} (q',w',u').$$

We say that $M$ {\em halts} on $x \in \Sigma^\ast$ if the
configuration $(q_0,\varepsilon,x)$ yields a configuration $(q_1,w,u)$
for some $\Sigma$-strings $w$ and $u$.
The number of steps $M$
takes to stop on a $\Sigma$-string $x$ is denoted by $T_M(x)$. If $M$ does not halt on $x \in \Sigma^\ast$ then we put
$T_M(x) =  \infty$.

The
{\em halting problem} for $M$ is an algorithmic question to
determine whether $M$ halts or not on  an input
$x\in \Sigma^\ast$, i.e., whether $T_M(x) = \infty$ or not.

We say that a TM $M$ {\em solves} or {\em decides} a decision
problem $D$ over an alphabet $\Sigma$ if $M$ stops on every input
$x \in \Sigma^\ast$ with an answer:
\begin{itemize}
    \item
$Yes$ (i.e., at a configuration $(f,\varepsilon,w)$, where $w$ starts with $a_1 a_1$) if $x \in
L(D)$;
    \item
$No$ (i.e., at configuration $(f,\varepsilon,w)$, where $w$ starts with $a_1 a_0$) otherwise.
\end{itemize}
We say that $M$ {\em partially decides} $D$ if it decides $D$
correctly on a subset $D'$ of $D$ and on $D - D'$ it either does
not stop or stops with an answer $DontKnow$ (i.e., stops at
configuration $(f,\varepsilon,w)$, where $w$ starts with $a_0$).
In the event when $M$ breaks or outputs $DontKnow$ the value of $T_M(x)$
is $\infty$.

\subsection{Polynomial time reductions}

For a function $f: \mathbb{N} \to \mathbb{N}$ define $\cclass{TIME}(f)$  [$\cclass{NTIME}(f)$ resp.]
to be the class of all decision problems decidable by some
deterministic [nondeterministic resp.] Turing machine within time $f(n)$.
Two of the most used classes of decision problems $\cclass{P}$ and
$\cclass{NP}$ are defined as follows:
    $$\cclass{P} = \bigcup_{k=1}^\infty \cclass{TIME}(n^k) ~\mbox{ and }~ \cclass{NP} = \bigcup_{k=1}^\infty \cclass{NTIME}(n^k).$$
Clearly $\P \subseteq \NP$. It is an old, open problem  whether $\NP = \P$ or not.

The classical {\em polynomial time many-to-one} or  {\em Karp  reductions} provide a  crucial tool to deal with problems in $\NP$. We recall it in the following definition and refer to them simply as to Ptime reductions.

\begin{definition}
Let $D_1$ and $D_2$ be decision problems. We say that a function
$f:\Sigma_{D_1}^\ast \rightarrow \Sigma_{D_2}^\ast$ is a {\em Ptime reduction}, or $f$ {\em Ptime reduces} $D_1$ to $D_2$, and write
$D_1 \stackrel{f}{\rightarrow}_P D_2$, if
\begin{itemize}
    \item
$f$ is polynomial time computable;
    \item
$x\in D_1$ if and only if $f(x)\in D_2$.
\end{itemize}
\end{definition}

We say that a Ptime reduction $f$ is {\em size-invariant} if
$$|x_1|<|x_2| \ \Longleftrightarrow \ |f(x_1)|<|f(x_2)|.$$
Notice, that many classical Ptime reductions are size-invariant
(see \cite{papa}).

Now, for a size-invariant reduction $f$ the function
    $$\CS_f(n) := |f(x)|, \ where \ |x| = n,$$
 is well defined and strictly increasing.  We refer to $\CS_f$ as the {\em size
growth} of $f$.

A problem $D \in \NP$ is called {\em $\NP$-complete} if every problem $D^\prime \in \NP$ is Ptime reducible to $D$.
The following is a classic result in complexity theory (see \cite{papa}).

\begin{theorem} \label{th:NP_compl}
The following holds.
\begin{enumerate}
    \item[(a)]
If $f$ is a Ptime reduction from $D_1$ to $D_2$ and $M$ is an Turing
machine solving $D_2$ in polynomial time then $M\circ f$ solves $D_1$ in polynomial time.
    \item[(b)]
3SAT is $\NP$-complete.
\end{enumerate}
\end{theorem}
Here, and below, by $M\circ f$ we denote the algorithm that is a composition of the TM $M$
and a TM that computes $f$.

\section{Distributional problems and generic case complexity}
\label{se:generic}

Let us first recall some basic definitions of probability theory that will be used in this section.
A \emph{probability measure} on $\Sigma^\ast$ is a function $\mu:\Sigma^\ast\to [0,1]$
satisfying $\sum_{x\in\Sigma^\ast} \mu(x)=1$.
An \emph{ensemble of probability measures} on $\Sigma^\ast$ is a
collection of sets $\{S_n\}_{n=1}^\infty$ of $\Sigma^\ast$ (not necessarily disjoint) and a collection
of probability measures $\mu=\{\mu_n\}_{n=1}^\infty$ satisfying $\supp(\mu_n) \subseteq S_n$
and $S=\bigcup S_n$.
A \emph{spherical ensemble of probability measures} on $\Sigma^\ast$ is an
ensemble with $S_n=\Sigma^n$. In particular,
a spherical ensemble of probability measures on $\Sigma^\ast$
is uniquely defined by a collection of measures $\{\mu_n\}_{n=1}^\infty$
satisfying $\supp(\mu_n) \subseteq \Sigma^n$.

\subsection{Distributional decision problems}
\label{subsec:dist}

The average case complexity deals with ``expected'' running time of algorithms, while the generic case complexity deals with ``most typical'' or {\em generic} inputs  of a given problem $D = (\Sigma^\ast,D^+)$.
These  require to measure or  compare various subsets of inputs from $\Sigma^\ast$.
 There are several standard ways to do so, for example, by introducing either a probability measure $\mu$ on $\Sigma^\ast$ (as was done in \cite{Levin:1986,Gurevich:1991}), or an {\em ensemble} of probability measures defined on spheres or balls of $\Sigma^\ast$ (see \cite{Impagliazzio95}).
 In many cases, all three approaches are equivalent and lead to similar results. Following the current tradition in  computer science, we elect here to work with a {\em spherical} ensemble $\mu = \{\mu_n\}_{n=1}^\infty$ of   probability measures $\mu_n$ defined on the spheres $\Sigma^n$.  In what follows, we always assume that $\Sigma$ is a finite alphabet and every  measure $\mu_n$ from the ensemble $\mu$ is atomic, i.e., it is given by a probability function (which we denote again by $\mu_n$) $\mu_n:\Sigma^n \to \mathbb{R}$ so that $\mu_n(S) = \sum_{x \in S} \mu_n(x)$ for every subset $S \subseteq \Sigma^n$. The pair $(\Sigma^\ast, \mu)$ is termed a {\em distributional space}. Whether $\mu_n$ is a probability measure or the corresponding probability function will be always clear from the context, so no confusion should arise.

We want to stress here that generic properties of a given  decision problem depend  on the chosen ensemble $\mu$ and $\mu$ is an essential part of the problem (see \cite{Gurevich:1991} for details).

\begin{definition}
A {\em distributional decision problem} is a triple $(\Sigma^\ast, D^+, \mu)$, where $D = (\Sigma^\ast, D^+)$
is a decision problem and $(\Sigma^\ast, \mu)$  a distributional space.
\end{definition}
  Usually  we refer to a distributional problem $(\Sigma^\ast, D^+, \mu)$ as a pair $(D, \mu)$, where $D = (\Sigma^\ast, D^+)$.

There are two important constructions on distributional spaces, introduced in \cite{Gurevich:1991}.
Since we use here ensembles of distributions,  unlike \cite{Gurevich:1991}, where single  measures were  used,  we give below precise definitions. Notice,  that we always assume that $\Sigma_{i \in J} a_i = 0$ if $J = \emptyset$.

\begin{definition}[Transfers of ensembles]
Let $\Sigma$ and $\Pi$ be finite alphabets,  $(\Sigma^\ast, \mu)$ and  $(\Pi^\ast, \nu)$ distributional spaces, and  $f:\Sigma^\ast \rightarrow \Pi^\ast$ a size-invariant function. Then  $\nu$ is the  $f$-{\em transfer} of  $\mu$ (or $f$ transfers   $\mu$ to $\nu$) if for any $y\in \Pi^\ast$ the following equality holds
\begin{equation}\label{eq:induced_f}
    \nu_{|y|}(y) =
    \left\{
    \begin{array}{ll}
    \sum_{x \in f^{-1}(y)} \mu_{|x|}(x), & \mbox{if $|y| = |f(z)|$ for some $z$};\\
    |\Pi|^{-|y|}, & \mbox{otherwise.}\\
    \end{array}
    \right.
\end{equation}
\end{definition}

\begin{definition} [Induced ensembles]
Let $(\Sigma^\ast, \mu)$ be a  distributional space and  $S \subseteq \Sigma^\ast$.  Then  an ensemble $\mu^S = \{\mu^S_n\}_{n=1}^\infty$ on   $\Sigma^\ast$ is called $S$-induced by $\mu$   if for any $x\in \Sigma^\ast$ the following equality holds
\begin{equation}\label{eq:induced_S}
    \mu^S_{|x|}(x) =
    \left\{
    \begin{array}{ll}
    \frac{\mu_{|x|}(\{x\}\cap S)}{\mu_{|x|}(S\cap\Sigma^{|x|})}, & \mbox{if }\mu_{|x|}(S\cap\Sigma^{|x|}) \ne 0;\\
    \mu_{|x|}(x), & \mbox{otherwise.}\\
    \end{array}
    \right.
\end{equation}
\end{definition}

\subsection{Generic complexity}

Let $(\Sigma^\ast, \mu)$ be a distributional space and $S \subset \Sigma^\ast$. The  function
$$n \mapsto \mu_n(S \cap \Sigma^n)$$
is called the {\em density function} of $S$  and its limit (if exists)
    $$\rho(S) = \lim_{n\rightarrow \infty}\mu_n(S \cap \Sigma^n)$$
is called the {\em asymptotic density} of $S$ in  $(\Sigma^\ast, \mu)$.

\begin{definition}
A subset $S \subseteq \Sigma^\ast$ is called
\begin{itemize}
    \item
{\em generic} in $\Sigma^\ast$ if $\rho(S) = 1$;
    \item
{\em strongly generic} in $\Sigma^\ast$ if $\rho(S) = 1$ and $\mu_n(S \cap \Sigma^n)$
converges to $1$ super polynomially fast, i.e.,
$ |1-\mu_n(S \cap \Sigma^n)| = O(n^{-k})$
for any $k \in \mathbb{N}$;
    \item
{\em negligible} in $\Sigma^\ast$ if $\rho(S) = 0$;
    \item
{\em strongly negligible} in $\Sigma^\ast$ if $\rho(S) = 0$ and $\mu_n(S \cap \Sigma^n)$
converges to $0$ super polynomially fast, i.e.,
$ |\mu_n(S \cap \Sigma^n)| = O(n^{-k})$ for any $k \in \mathbb{N}$.
\end{itemize}
\end{definition}

Notice that we use the term ``generic'' in the sense of ``typical''.
The same term has also been used in complexity and set theory
to refer to sets that are far from typical, that are constructed
through Cohen forcing.

\begin{definition}\label{de:generic_poly}
Let $(D,\mu)$ be a distributional decision problem.
\begin{itemize}
    \item
We say that $(D,\mu)$ is {\em decidable generically in polynomial time} (or {\em GPtime
decidable}) if there exists a Turing machine $M$ that partially decides $D$
within time $T_M(x)$ and a polynomial $p(x)$ such that
    $$\mu_n\{x \in \Sigma^n \mid T_M(x)>p(n) \} = o(1).$$
In this case we say that $M$ is a generic polynomial time decision algorithm for $D$ and $D$ has {\em generic time
complexity} at most $p(n)$.
    \item
We say that $(D,\mu)$ is {\em decidable strongly generically in polynomial time} (or {\em
SGPtime decidable}) if there exists a Turing machine $M$ that partially decides
$D$ within time $T_M(x)$ and a polynomial $p(x)$ such that for any
polynomial $q(n)$
    $$\mu_n\{x\in \Sigma^n \mid T_M(x)>p(n) \} = o(1/q(n)).$$
In this case, we say that $M$ is a strongly generic polynomial time decision algorithm for $D$ and $D$ has {\em strong generic time
complexity} at most $p(n)$.
\end{itemize}
We refer to the sequence $\mu_n\{x\in \Sigma^n \mid T_M(x)>p(n) \}$ as a {\em control sequence} of the
algorithm $M$ relative to the complexity bound $p$ and denote it by $\CC_{M,p}$.
\end{definition}

In other words, a problem $(D,\mu)$ is GPtime (SGPtime) decidable if there exists a polynomial time TM that partially decides $D$ and its halting set is generic (strongly generic) in $(\Sigma^\ast,\mu)$.

\subsection{Distributional $\NP$-problems}
 \label{subsec:DistNP}

In this section we recall the notion of a distributional $\NP$-problem,
which is a distributional analog of the
classical $\NP$-problems.

\begin{definition}
[Ptime computable real-valued function]
A function $f:\Sigma^\ast \rightarrow [0,1]$ is {\em computable in polynomial time}
if there exists a polynomial time algorithm that for every $x\in \Sigma^\ast$
and $k\in\MN$ computes a binary fraction $f_{x,k}$ satisfying
    $$|f(x) - f_{x,k}| < 2^{-k}.$$
\end{definition}

\begin{definition}
[Ptime computable ensembles of probability measures]
We say that a spherical ensemble of measures $\mu = \{\mu\}_{n=1}^\infty$ on $\Sigma^\ast$ is Ptime
computable if the function $\Sigma^\ast \to [0,1]$ defined by $x \to \mu_{|x|}(x)$ is Ptime computable.
\end{definition}

Denote by $\mu^\ast = \{\mu_n^\ast\}_{n=1}^\infty$ the ensemble of probability {\em distributions} defined by
    $$\mu_{|x|}^\ast(x) = \mu_{|x|}\rb{\{y\in\Sigma^{|x|} \mid y <_{slex} x\}}.$$
 As above, the ensemble $\mu^\ast$ is called Ptime computable if the function $x \to \mu_{|x|}^\ast(x)$ is Ptime computable.

\begin{lemma}\label{le:S-induced}
Let $(\Sigma^\ast, \mu)$ be a distributional space. Then the following hold:
\begin{itemize}
\item[(a)]
If   $\mu^\ast$ is  Ptime computable then $\mu$ is   Ptime computable.
\item[(b)]
If $S$ is a subset of $\Sigma^\ast$ such that the function $n \to \mu_n(S \cap \Sigma^n)$ is Ptime computable then the $S$-induced on $\Sigma^\ast$ ensemble of measures $\mu^S$   is Ptime computable.
\end{itemize}
\end{lemma}
  \begin{proof}
  Follows directly from definitions.
  \end{proof}

\begin{definition}
$\DistNP$ is a class of distributional decision problems $(D,\mu)$ such
that
\begin{itemize}
    \item
$D \in \NP$;
    \item
$\mu^*$ is a Ptime computable ensemble of probability distributions  on $\Sigma_D^\ast$.
\end{itemize}
\end{definition}

\begin{definition}
$\GP$ is the class of GPtime decidable distributional decision problems
(not necessarily from $\DistNP$). $\SGP$ is the class of SGPtime decidable
distributional decision problems.
\end{definition}

We want to point out that classes $\GP$ and $\SGP$ contain some exotic
problems, e.g., some undecidable problems. For more information see
\cite{HM,GMMU:2007,MRyb}.

\section{Generic Ptime reductions}
\label{se:reductions}

In this section we introduce the notion of a {\em generic polynomial
reduction}  and describe two particular types of
reductions, called size and measure reductions.

Observe first that the classical Karp reductions do not work for generic complexity. Indeed, the following example shows that   a
Ptime reduction $D \stackrel{f}{\rightarrow} E$ and a generic polynomial time decision algorithm for $E$ do not immediately provide a generic polynomial time decision algorithm for $D$.

\begin{example}
\label{ex:gen-red}
Let $\Sigma = \{0,1\}$ be a binary alphabet and $\mu$
the spherical ensemble of uniform measures $\mu_n$  on $\Sigma^n$. Let
$f:\Sigma^\ast \rightarrow \Sigma^\ast$ be a monoid homomorphism defined by
    $$0 \stackrel{f}{\mapsto} 00 ~~\mbox{ and }~~ 1 \stackrel{f}{\mapsto} 1.$$
Now, for a decision problem $D = (\Sigma^\ast,D^+)$ consider a decision problem $f(D) = (\Sigma^\ast,f(D^+))$.
It follows from the construction that $D  \stackrel{f}{\rightarrow} f(D)$
is a Ptime reduction and $f(D) \in \NP$, provided $D \in \NP$. Furthermore,  it is easy to check that the set
$f(\Sigma^\ast)$, as well as $f(D^+)$, is strongly negligible in $(\Sigma^\ast, \mu)$.
This implies that a partial  algorithm $A$ that on an each input from $\Sigma^\ast \setminus f(\Sigma^\ast)$ says ``No'' and does halt on $f(\Sigma^\ast)$, is a strongly generic polynomial time decision algorithm for  $(f(D), \mu)$.  Nevertheless, $A$ does not reveal any useful information on $D$.
\end{example}

\begin{definition}
Let $(\Sigma^\ast,D,\mu), (\Delta^\ast,E,\nu) \in \DistNP$ and
$D  \stackrel{f}{\rightarrow} E$ a Ptime size-invariant reduction.
\begin{itemize}
    \item[{\bf(R0)}]
We say that $f$ is a {\em weak GPtime reduction} if there exists a TM $M$,
which GPtime decides $(E,\nu)$ and $M \circ f$ GPtime decides
$(D,\mu)$.

\item [{\bf(R1)}] We say that $f$ is an {\em GPtime reduction} if for every TM $M$,
which GPtime decides $(E,\nu)$ the composition $M \circ f$ GPtime
decides $(D,\mu)$.
    \item[{\bf(R2)}]
We say that $f$ is an {\em SGPtime reduction} if for every TM $M$, which
SGPtime decides $(E,\nu)$ the composition $M \circ f$ SGPtime
decides $(D,\mu)$.
\end{itemize}
\end{definition}

We give examples of SGPtime reductions in the next two sections.

\begin{remark}
One can introduce reductions $D  \stackrel{f}{\rightarrow} E$ of a more general type by allowing the function $f$ to be defined only on a generic (strongly generic) subset $Y$ of $\Sigma_D^\ast$ with the polynomial time computable characteristic function $\chi_Y$.
\end{remark}

\begin{proposition}[Transitivity of GPtime and SGPtime reductions]
The classes of all GPtime and SGPtime reductions are closed under composition.
\end{proposition}

\begin{proof}
Follows directly from the definitions.
\end{proof}

It is not known if the class of weak GPtime reductions is transitive.

\begin{definition}
Let $(D,\mu)$ be a distributional decision problem. We say that
  \begin{itemize}
  \item $(D,\mu)$ is {\em SGPtime hard}
for $\DistNP$ if every $\DistNP$ problem SGPtime reduces to $(D,\mu)$.
  \item $(D,\mu)$  is {\em SGPtime complete} for $\DistNP$ if $(D,\mu) \in \DistNP$ and $(D,\mu)$ is {\em SGPtime hard} for $\DistNP$.
\end{itemize}
\end{definition}

\subsection{Change of size}
\label{se:change_size}

In this section, we introduce {\em change of size} (CS) reductions.

\begin{definition}
\label{de:CS}
Let $(D,\mu), (E,\nu) \in \DistNP$ and $D  \stackrel{f}{\rightarrow} E$
a Ptime size-invariant reduction of $D$ to $E$. If $\nu$ is the $f$-transfer by $\mu$ (see Section \ref{subsec:dist}) then   $f$ is called a {\em CS-reduction}.
\end{definition}


\begin{figure}[htbp]
\centerline{
\includegraphics[scale=0.6]{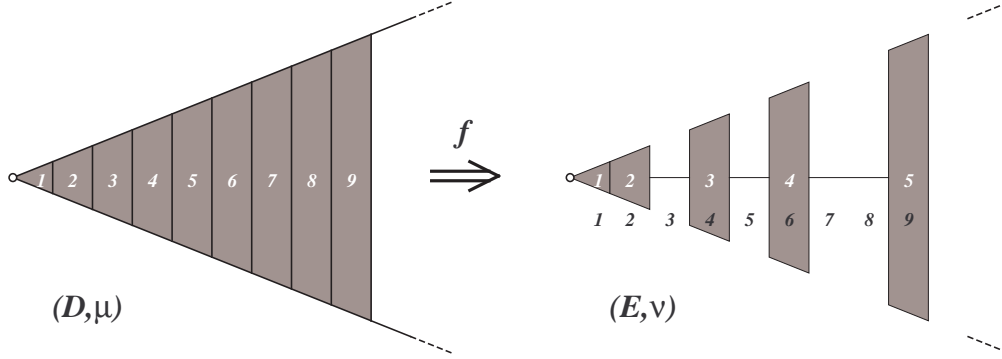} }
\caption{\label{fi:size_increase} In this example $(E,\nu)$ is obtained from $(D,\mu)$ by
increasing the sizes of elements.}
\end{figure}

\begin{theorem}[CS-reductions are  GPtime and SGPtime reductions]\label{le:CS}
Let $(D,\mu), (E,\nu) \in \DistNP$ and
$(D,\mu)  \stackrel{f}{\rightarrow} (E,\nu)$ a CS-reduction. If $\CS_f$ is bounded by a
polynomial, then $f$ is a GPtime and SGPtime reduction.
\end{theorem}

\begin{proof}
Let  $A$ be an algorithm that generically decides $(E,\nu)$ within
a polynomial time upper bound  $p(m)$. Then $A \circ f$  is a partial decision algorithm for $(D,\mu)$. Since $\nu$ is induced by
$\mu$, one has:
\begin{align*}
o(1) & = \nu_{\CS_f(k)}\{f(x) \mid T_A(f(x))>p(\CS_f(k)), ~ |f(x)|=\CS_f(k)\} \\
&= \mu_{k}\{x \mid T_{A} (f(x))>p(\CS_f(k)), ~ |x|=k\} \\
&\leq  \mu_{k}\{x \mid T_{A\circ f} (x)>p(\CS_f(k))+ T_f(k), ~ |x|=k\}.
\end{align*}
Observe, that $p \circ \CS_f  + T_f$  is polynomially bounded, since  $\CS_f$ and $T_f$ are polynomially  bounded. Clearly, the control sequence
$\CC_{A\circ f,p \circ \CS_f + T_f}$ is at most $\CC_{A,p} \circ  \CS_f$.  Notice, that
$\CC_{A,p} \circ  \CS_f$ is an infinite  subsequence of $\CC_{A,p}$ (because  $\CS_f$ is strictly increasing), hence it converges to $0$, so $p \circ \CS_f + T_f$ is a generic upper bound for $A \circ f$. This proves the first statement of the theorem.

To prove the second statement, assume that  $(E,\nu)$ is SGPtime decidable by $A$ within
a polynomial time $p$. Then for the control sequence $\CC_{A,p}$ one has
$$
\CC_{A,p} = o(1/n^k)
$$
for any positive integer $k$. Due to the  inequalities above,
the control sequence for $A\circ f$  with respect to the polynomial bound
$p \circ \CS_f + T_f$ satisfies the following inequality
$$
\CC_{A\circ f,p \circ \CS_f + T_f} = o(1/\CS_f(n)^k) \le o(1/n^k)).
$$
Hence $(D,\{\mu_n\}_{n=1}^\infty)$ is SGPtime decidable by $A\circ f$, as claimed.
\end{proof}

By Theorem \ref{le:CS} a CS-reduction generally increases
time complexity and improves control sequence.

\subsection{Change of measure}
\label{se:change_measure}

In this section, we define {\em change of measure} (CM) reductions.

\begin{definition}
\label{de:CM-red}
Let $(D,\mu), (E,\nu) \in \DistNP$ and $D  \stackrel{f}{\rightarrow} E$ a
Ptime reduction  such that
\begin{itemize}
    \item
$|x| = |f(x)|$ for any $x \in \Sigma_D^\ast$;
    \item
there exists a polynomial $d$ such that for each $x \in \Sigma_D^\ast$
$$\nu_{|x|}(f(x)) \ge \frac{\mu_{|x|}(x)}{d(|x|)}.$$
\end{itemize}
Then $f$ is called a {\em CM-reduction}.
\end{definition}

Figure \ref{fi:measure_reduction} depicts the situation under consideration.

\begin{figure}[htbp]
\centerline{
\includegraphics[scale=0.6]{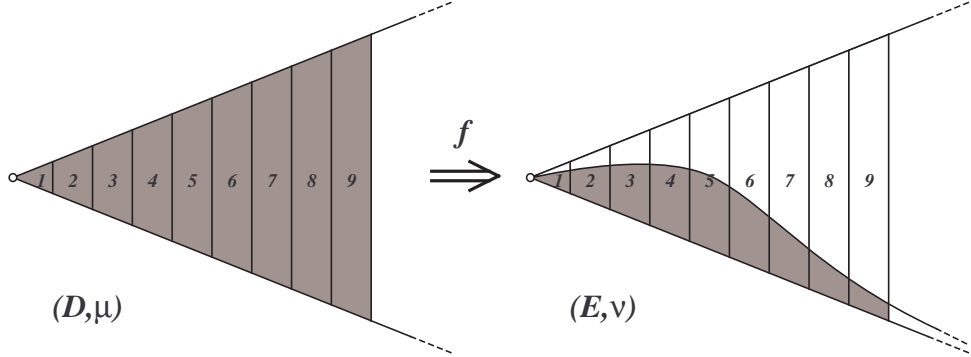} }
\caption{\label{fi:measure_reduction} Scheme of a CM-reduction. A
function $f$ maps a distributional decision problem $(D,\mu)$
into a distributional decision problem $(E,\nu)$ so that the $i$th
sphere in $D$ is mapped exactly into the $i$th sphere in $E$. The
grey part of $E$ depicts the image of $D$.}
\end{figure}

\begin{theorem}[CM-reduction is an SGPtime reduction]\label{le:CM}
Let $(D,\mu), (E,\nu) \in \DistNP$ and
$D   \stackrel{f}{\rightarrow} E$ a CM-reduction. Then the following holds.
\begin{itemize}
\item[(a)]
If $(E,\nu)$ is decidable by a TM $A$ within a  generic polynomial time
bound  $p$ and $\CC_{A,p} = o(1/d(k))$ (where $d(k)$ is the function from Definition \ref{de:CM-red}) then $A \circ f$ GPtime decides $(D,\mu)$.
\item[(b)]
$f$ is an SGPtime reduction.
\end{itemize}
\end{theorem}

\begin{proof}
(a) Let  $A$ be an algorithm that generically decides $(E,\nu)$ within
a polynomial time upper bound  $p(m)$. Then $A \circ f$  is a partial decision algorithm for $(D,\mu)$. Recall, that $f$ preserves the size.
Therefore, $A \circ f$ decides $D$ within the polynomial time bound $p + T_f$ everywhere, except, maybe, a subset
    $$\{x \in D \mid T_A(f(x)) > p(|f(x)|) = p(|x|)\}.$$
To prove the statement it suffices to show that the set above is generic in $(\Sigma_D^\ast, \mu)$.
\begin{align*}
& \mu_k\{ x \in D \mid T_A(f(x)) > p(|f(x)|),~ |x|=k\} \\
\le\;\; & \nu_k\{f(x) \in E \mid T_A(f(x)) > p(|x|),~ |x|=k \} d(k) \\
\le\;\; & \nu_k\{y \in E \mid T_A(y) > p(|y|),~ |y|=k \} d(k)\\
 =\;\; &\CC_{A,p}(k) d(k) = o(1).
\end{align*}
(b) If an algorithm $A$ SGPtime decides $(E,\nu)$ then $\CC_{A,p} = o(1/q(k))$
for any polynomial $q$. Therefore, by part 1), $\CC_{A \circ f,p+T_f} \le o(d(k)/q(k))$ for every polynomial $q$. In particular, for $q = dq^\prime$ one has
$$\CC_{A \circ f,p+T_f} \le o(d(k)/q(k)) = o(1/q^\prime(k))$$
for any polynomial $q^\prime$, as required.
\end{proof}

\subsection{Reduction to a problem with the binary alphabet}
\label{subsec:binary}

In this section we show that each $\DistNP$ problem over a finite alphabet
$\Sigma$ can be reduced to a $\DistNP$ problem over a binary alphabet
$\{0,1\}$.

\begin{theorem}\label{th:bin_alph}
Let $(D,\mu)$ be an $\DistNP$ problem over a finite alphabet $\Sigma$. Then
there exists a $\DistNP$ problem $(E,\nu)$ over the binary alphabet
$\{0,1\}$ and a CS-reduction $D  \stackrel{f}{\rightarrow} E$ with linear size
function $\CS_f$.
\end{theorem}

\begin{proof}
Suppose that $\Sigma = \{a\}$ is an one-letter  alphabet.
Let $f:\{a\}^\ast \to \{0,1\}^\ast$ be a monoid homomorphism defined by $f(a) = 0$.
Put $E^+ = f(D^+)$ and $E = (\{0,1\}^\ast, E^+)$.
Define a spherical  ensemble of measures $\nu$ on
$\{0,1\}^\ast$ to be $$\nu_{|y|}(y) = \sum_{f(x)=y} \mu_{|x|}(x).$$
 Clearly,
$(E,\nu) \in \DistNP$ and $f$ is a
CS-reduction with linear size-growth function $\CS_f$.
By Theorem  \ref{le:CS}, $f$ is an SGPtime reduction.

Suppose that $|\Sigma| = n$, where $n \ge 3$. Define a function $f$ as
follows. Put $f(\varepsilon) = \varepsilon$ and, if $|x| \ge 1$ and $x$
is the $k$th element in $\Sigma^n$ (in the lexicographical order), then
$f$ maps $x$ into the $k$th element of $\{0,1\}^{\lceil |x| \log_2 n
\rceil}$. As above, we put $E^+ = f(D^+)$.
Let $\nu$ be the  $f$-transfer of $\mu$.
The problem $(E,\nu)$ belongs to
$\DistNP$ because $(D,\mu) \in \DistNP$ and $f$ is a Ptime reduction.   Clearly, $f$ is
a CS-reduction with a linear
size-growth function $\CS_f(i) = \lceil i \log_2 n \rceil$.
By Theorem  \ref{le:CS}, $f$ is an SGPtime reduction.
\end{proof}

\subsection{On restrictions of problems}
\label{subsec:restrictions}

Let $D = (\Sigma_D^\ast,D^+)$ be a problem and $S \subseteq \Sigma_D^\ast$.  In this section we consider the restriction  $D_S$ of $D$ to the subset $S$. Intuitively,  $D_S$  is the same problem as $D$, only the set of inputs is restricted to $S$.  The most natural formalization of $D_S$ would be $(S,D^+ \cap S)$, allowing the domain $S$ not equal to $\Sigma_D^\ast$, contrary to our assumption on algorithmic problems. In this case one can stratify the domain $I$ as a union $I = \cup_{n = 0}^\infty I_n$, where $I_n = I \cap \Sigma^n$, and leave only those $I_n$ that are non-empty. Then, one can obtain an ensemble of measures $\mu^\prime = \{\mu_n^\prime\}_{n = 0}^\infty$ on $I$ relative to the stratification above, where  $\mu_n^\prime$  is  the  measure  on $I_n$ induced by $\mu_n$. After that, the theory of distributional problems of this type can be developed similarly to the one already considered. However, it is a bit awkward and heavier in notation. We choose another way around this problem -- we change the ensemble of measures, but do not change the input space.

 Let $(D,\mu)$
be a distributional problem. For a subset $S\subseteq \Sigma_D^\ast$ consider
the   ensemble of probability measures
$\mu^S$ on $\Sigma_D^\ast$ $S$-induced by $\mu$ (see Section \ref{subsec:dist}).
The distributional problem $(D,\mu^S)$ is called the {\em restriction} of the distributional problem $(D,\mu)$ to the subset $S$.

\begin{lemma}
\label{le:comp-induced}
Let  $(D,\mu)  \in \DistNP$ and $S \subseteq \Sigma_D^\ast$.  If the function
$n \to \mu_n(S \cap \Sigma_D^n)$ is Ptime computable then $(D,\mu^S) \in \DistNP$.
\end{lemma}
  \begin{proof}
  Follows immediately from Lemma \ref{le:S-induced}.
\end{proof}

\begin{lemma}
\label{le:GP-S-induced}
Let $(D,\mu) \in \DistNP$, $S \subseteq \Sigma_D^\ast$,  and $(D,\mu^S) \in \DistNP$.
If an algorithm $A$ GPtime
decides $(D,\mu)$ with a control sequence $q_i$ such that  the sequence
$$
c_i= \left\{
\begin{array}{ll}
q_i/\mu_i(S \cap \Sigma_D^i), & \mbox{if } \mu_i(S \cap \Sigma_D^i)\ne 0;\\
q_i, & \mbox{if } \mu_i(S \cap \Sigma_D^i)=0.\\
\end{array}
\right.
$$
converges to $0$, then $A$ GPtime decides $(D,\mu^S)$ with the control
sequence bounded from above by  $\{c_i\}_{i=1}^\infty$.
\end{lemma}

\begin{proof}
Let $p(n)$ be a generic polynomial time upper bound of the algorithm $A$ and
$F = \{x\in\Sigma_D^\ast \mid T_A(x)>p(|x|)\}$.
Set $S_i = S \cap \Sigma_D^i, F_i = F \cap \Sigma_D^i$. Then
 $q_i = \mu_i(F_i)$ and for $i\in\MN$ one has
    $$\mu_i^S(F_i) =
    \left\{
    \begin{array}{ll}
    \frac{\mu_i(F_i\cap S_i)}{\mu_i(S_i)}, & \mbox{if } \mu_i(S_i)\ne 0;\\
    \mu_i(F_i), & \mbox{if } \mu_i(S_i)= 0.\\
    \end{array}
    \right.
    $$
Hence
    $$\mu_i^S(F_i) \le
    \left\{
    \begin{array}{ll}
    \frac{q_i}{\mu_i(S_i)}, & \mbox{if } \mu_i(S_i)\ne 0;\\
    q_i, & \mbox{if } \mu_i(S_i)= 0.\\
    \end{array}
    \right.
    $$
Thus, the sequence $\left\{\mu_i^S(F_i)\right\}_{i=0}^\infty$ converges to $0$.
\end{proof}

\begin{corollary}\label{co:restriction}
Let $(D,\mu) \in \DistNP$, $S \subseteq \Sigma_D^\ast$,  and $(D,\mu^S) \in \DistNP$.
If there exists a polynomial $d$ such that $\mu_i(S \cap \Sigma_D^i) \ge 1/d(i)$ for $\mu_i(S \cap \Sigma_D^i) \neq 0$
then the identity function $id:\Sigma_D^\ast \hookrightarrow \Sigma_D^\ast$ gives an SGPtime
reduction
$$(D,\mu^S) \stackrel{id}{\to} (D,\mu).$$

\end{corollary}

\begin{remark}
We would like to emphasize that
the situation with restrictions of problems in $\GP$ is quite different from the ``average-case'' one,
where almost any  restriction preserves the property of being polynomial time computable on average.
\end{remark}

\section{Distributional bounded  halting problem}
\label{se:DistHP}

In this section we, following \cite{Gurevich:1991}, define the distributional  bounded halting problem
and prove that it is SGPtime complete in $\DistNP$.

Let $M$ be a nondeterministic Turing machine with the binary tape
alphabet $\Sigma = \{0,1\}$. Intuitively, the bounded halting problem for $M$
is the following algorithmic question:
\begin{quote}
For a positive integer $n$ and a binary string $w$ such that $|w|
< n$ decide if there is a halting computation for $M$ on $w$ in
at most  $n$ steps.
\end{quote}
By our definitions (see Section \ref{se:prelim}) instances  of algorithmic problems are words (not pairs of words) in some alphabet, so to this end we encode a pair $(n,w)$ by the binary string $c(n,w) = 1^m 0w$ such that $n = |1^m 0 w|$. Notice, that  any binary string
containing  $0$ is the code $c(n,w)$ for some $n, w$. Denote by  $BH(M)^+$  the subset of all binary strings $c(n,w)$, where $n \in \mathbb{N}, w \in \Sigma^\ast$,   such that $M$ halts on $w$ within $n$ steps. From now on we refer to the problem $BH(M) = (\Sigma^\ast, BH(M)^+)$ as the {\em bounded halting problem}.

To turn $BH(M)$ into a distributional
 problem we introduce a spherical  ensemble $\nu = \{\nu_n\}_{n = 1}^\infty$ of probability measures as follows. For $u\in \{0,1\}^\ast$
 put
    $$\nu_{|u|}(u) =
    \left\{
    \begin{array}{ll}
    \frac{1}{|u| 2^{|w|}}, & \mbox{if } u=1^m0w;\\
    1, & \mbox{if } u = \varepsilon;\\
    0, & \mbox{if $u = 1^k$ for some $k\ge 1$}.\\
    \end{array}
    \right.
    $$
The problem $(BH(M), \nu)$ is the {\em distributional} halting problem for $M$, we refer to it as $DBH(M)$.

\begin{figure}[htbp]
\centerline{ \includegraphics[scale=0.5]{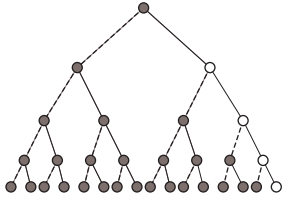} }
\caption{\label{fi:space} Probability space for halting problem (a
finite part). Dashed lines correspond to $0$, solid lines
correspond to $1$. Each dot is associated with the label of the
path from the root to itself. Grey dots have non-trivial measure.}
\end{figure}

A positive integer $n$ is called {\em longevous} for an input $w$
of an NTM $M$ if every halting computation of $M$ on $w$ has at
most $n$ steps. A function $g(n)$ is a {\em longevity guard} for
$M$ if for every input $w$ the number $g(|w|)$ is longevous for
$w$. Notice, that if $g$ is a longevity guard for $M$, then any
function $h\ge g$ is also a longevity guard for $M$. In what follows, we always assume that
a longevity guard satisfies the following conditions:
\begin{itemize}
\item[{\bf(L1)}]
$g(|w|) \ge |w|$;
\item[{\bf(L2)}]
$g(|w|)$ is strictly increasing.
\end{itemize}

\begin{remark}
\label{re:longevity}
For every problem $D \in \NP$, there is an NTM   $D$ that decides $D$ and has a
polynomial longevity guard $g(n)$, satisfying the conditions (1), (2) above.
\end{remark}
  Since $M$ halts on an input $w$ if and only if it halts on $w$ within $g(|w|)$ steps, there is no much use to consider instances $(n,w)$ of the halting problem for $M$ with $n > g(|w|)$. A rigorous formalization of this observation is to restrict the problem $DBH(M)$ to the subset of instances
$$
C(g) = \{ c(g(|w|),w) \mid  w\in \Sigma^\ast \}.
$$

More generally, for a computable function $g(n)$, satisfying conditions (1) and (2), consider the set $C(g)$ as above and denote by $\nu = \nu(g)$ the ensemble of measures for $\Sigma^\ast$ which is $C(g)$-induced by $\nu$. Let  $DBH(M,g) = (BH(M), \nu_g)$ be the restriction of the problem $DBH(M)$ to $C(g)$.

\begin{proposition}\label{pr:restriction2BHMg}
Let $M$ be an NTM  and $g:\MN \to \MN$ a polynomial function. Then
$DBH(M,g) \in \DistNP$  and the identity function $id : \Sigma^\ast \to \Sigma^\ast$ gives an  SGPtime reduction
$DBH(M,g) \stackrel{id}{\to} DBH(M)$.
 \end{proposition}

\begin{proof}
Observe first that the function $n \to \nu_n(C(g) \cap \Sigma^n)$ is Ptime computable. Indeed,
if $u = 1^m0w \in C(g) \cap \Sigma^n$, then $n = g(|w|)$, so $|w| = g^{-1}(n) = k$ is uniquely defined (since $g$ is monotone). In this case, $\nu_n(u) = \frac{1}{n 2^{k}}$ depends only on $n$, hence
 $$\nu_n(C(g) \cap \Sigma^n) = \frac{1}{n 2^{k}}|C(g) \cap \Sigma^n| = \frac{1}{n 2^{k}} 2^k = \frac{1}{n}.$$ Therefore,
\begin{equation}
\label{eq:nu-C(g)}
\nu_n(C(g) \cap \Sigma^n) =
    \left\{
    \begin{array}{ll}
    \frac{1}{n}, & \mbox{if } g^{-1}(n) \neq \emptyset;\\
    0, & \mbox{otherwise}.\\
    \end{array}
    \right.
\end{equation}
Since the function $g$ is polynomial it takes at most $O(ng(n))$ time  to check if  $g^{-1}(n) = \emptyset$ or not.
  Now, by Lemma \ref{le:comp-induced} $DBH(M,g) \in \DistNP$.
     Equalities \ref{eq:nu-C(g)} and Corollary \ref{co:restriction} imply that  $DBH(M,g) \stackrel{id}{\to} DBH(M)$ is an SGPtime reduction, as claimed.
\end{proof}


\begin{figure}[htbp]
\centerline{ \includegraphics[scale=0.5]{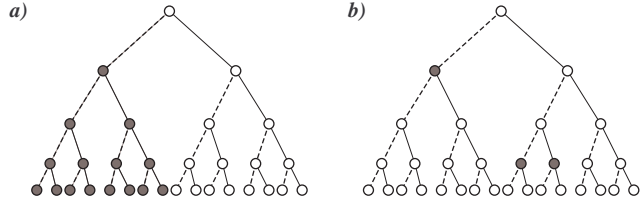} }
\caption{\label{fi:restriction} Examples of restricted problems
$BH(M,g)$, where a) $g(n) = n+1$ and b) $g(n) = 2n+1$. Only grey
dots have non-trivial probability.}
\end{figure}

In the proofs below,  we use the following
encoding of natural numbers and Turing machines. Let a string  $b = b_k \ldots b_0$ be a
 binary expansion for $n\in\mathbb{N}$, i.e., $n =
\sum_{i=0}^k b_i 2^i$, $b_k=1$, and $b_0, \ldots, b_{k-1} \in
\{0,1\}$. Denote by $\overline{n}$ the binary string
$$1 b_k 1 b_{k-1 } \ldots 1 b_0$$
obtained from $b$ by inserting $1$ in front of each
symbol in $b$.  Let $\gamma:M \to \gamma(M)$ be a polynomial time computable enumeration of (nondeterministic) Turing machines, such that $\gamma(M)$ is a binary representation of a natural number and every natural number is equal to $\gamma(M)$ for some $M$. Denote by $\overline{M}$ the string $\overline{\gamma(M)}$.

\begin{theorem}\label{th:red2BH}
For any $(D,\mu) \in \DistNP$ there exists an NTM $M$ over the  binary alphabet $\Sigma = \{0,1\}$, a
polynomial longevity guard $g$ for $M$, and an SGPtime reduction
$(D,\mu) \stackrel{f}{\rightarrow} DBH(M,g)$.
\end{theorem}

\begin{proof}
Fix $(D,\mu) \in \DistNP$. We divide the proof of the theorem into two parts.
First, we construct a Turing machine $M$, a longevity guard $g$ for
$M$, and a Ptime reduction $f$ from the original problem $(D,\mu)$
to $BH(M,g)$. Then we show that $f$ is a composition of CS and CM
reductions defined in Sections \ref{se:change_size} and
\ref{se:change_measure} and, hence, $f$ is an SGPtime reduction $(D,\mu) \stackrel{f}{\rightarrow} DBH(M,g)$.

{\bf Part I.}
By Theorem \ref{th:bin_alph} we may assume that the alphabet of
$D$ is binary. Now, since $D\in\NP$ there exists an NTM $A_D$
such that:
\begin{itemize}
    \item
$A_D$ has a halting computation on an input $w$ if and only if $w\in D$;
    \item
$A_D$ has a polynomially bounded longevity guard.
\end{itemize}

Recall that $\mu^\ast$ is the ensemble of probability distributions
    $$\mu_{|x|}^\ast(x) = \mu_{|x|}\rb{\{y\in\Sigma^{|x|} \mid y<_{lex}x\}}.$$
Since $(D,\mu) \in \DistNP$ the ensemble $\mu^\ast = \{\mu_n^\ast\}$ is Ptime computable.
 For $x\in\{0,1\}^\ast$ define a
function
    $$
    \hat\mu_{|x|}(x) =
    \left\{
    \begin{array}{ll}
    \mu_{|x|}^\ast(x^+), & \mbox{if } x\ne 1^{|x|};\\
    1, & \mbox{if } x=1^{|x|};\\
    \end{array}
    \right.
    $$
where $x^+$ is the  lexicographic successor of $x$.
For $x\in\{0,1\}^\ast$ such that $\mu_{|x|}(x) > 2^{-|x|}$ define
$x^\prime = x_0 \ldots x_k$ to be the smallest (in shortlex ordering)  binary string such that
    $$\mu_{|x|}^\ast(x) < x_0.x_1 x_2 \ldots x_k 1 \le \hat\mu_{|x|}(x),$$
where $x_0.x_1 x_2 \ldots x_k$ is the binary expansion of a real number in the interval $[0,1]$.
One can describe $x^\prime$ as follows. Assume first that $\hat\mu_{|x|}(x) \neq 1$.
Since $\mu_{|x|}(x) > 2^{-|x|}$ the binary expansions of $\mu^{\ast}_{|x|}(x)$ and
$\hat\mu_{|x|}(x)$ differ within the first $|x|-1$ bits after ``.'', i.e.,
    $$\mu^{\ast}_{|x|}(x) = 0.x_1x_2 \ldots x_k 0 \ldots ~~\mbox{ and }~~ \hat\mu_{|x|}(x) = 0.x_1x_2 \ldots x_k 1 \ldots$$
for some $k \leq |x|-1$.
In this case  $x^\prime = 0x_1x_2 \ldots x_k$. The case $\hat\mu_{|x|}(x) = 1$ is similar.
It follows that for every $x\in\Sigma^\ast$ such that $\mu_{|x|}(x)>2^{-|x|}$ we have $|x'| \le |x|$  and
    $$x_0.x_1 x_2 \ldots x_k 1 - 2^{-|x'|}
    \le
    \mu_{|x|}^\ast(x)
    <
    \hat\mu_{|x|}(x)
    <
    x_0.x_1 x_2 \ldots x_k 1 + 2^{-|x'|}.$$
Hence, $\mu_{|x|}(x) < 2\cdot 2^{-|x'|}$.  Define
$$
x'' =
\left\{
\begin{array}{ll}
0x,  & \mbox{if } \mu_{|x|}(x) \le 2^{-|x|};\\
1x', & \mbox{if } \mu_{|x|}(x) > 2^{-|x|}.\\
\end{array}
\right.
$$
Notice that for every $x\in\Sigma^\ast$, $\mu_{|x|}(x) \le 4\cdot 2^{-|x''|}$ and
$|x''| \le |x|+1$.

Now we describe  an NTM $M$, a function $g$,  and a reduction $f:(D,\mu) \rightarrow BH(M,g)$.
If $M$ is defined and  $g$ is a polynomial longevity guard for $M$, then the reduction $f$
is defined for $x \in \Sigma^\ast$ by
\begin{equation}\label{eq:red2BH}
    f(x) = 1^{g(|x|)-|0\overline{n}0x''|}0 \overline{n} 0x''
\end{equation}
where $n = |x|$. It is left to define $M$ and $g$.

 The machine $M$ on a binary  input $u$ executes the following algorithm:
\begin{enumerate}
    \item[A.]
If $u$ is not in the form $1^k 0 \overline{n} 0bw$, where $b\in\Sigma$ and $w \in \Sigma^\ast$, then loop forever.
    \item[B.]
If $u = 1^k 0 \overline{n} 0bw$ then decode $n$. 
    \item[C.]
If $b = 0$:
\begin{enumerate}
    \item[(1)]
if $\mu_n(w) > 2^{-|w|}$ then loop forever;
    \item[(2)]
otherwise simulate $A_D$ on $w$.
\end{enumerate}
    \item[D.]
If $b=1$:
\begin{enumerate}
    \item[(1)]
find the lexicographic smallest $x\in \{0,1\}^n$ satisfying $\mu_n^\ast(x) < 0.w1 \le \hat\mu_n(x)$
using divide and conquer approach;
    \item[(2)]
if $\mu_n(x) \le 2^{-|x|} $ or $x' \ne w$ then loop forever;
    \item[(3)]
otherwise simulate $A_D$ on $x$.
\end{enumerate}
\end{enumerate}
By construction, $M$ has a halting computation on $u\in\Sigma^\ast$ if and only if $u = f(x)$ for some $x\in\Sigma^\ast$
and $x\in D^+$.

We claim that $M$ has a polynomial longevity guard $g$.
Indeed, since $D\in\NP$ it follows that an NTM $A_D$ has a polynomial longevity guard,
and all steps in the algorithm above, except simulation of $A_D$,
can be performed by deterministic polynomial time algorithms. Therefore, $M$
has a polynomial longevity guard $g$, as claimed.
In particular, $D \stackrel{f}{\to} BH(M,g)$ is a Ptime reduction, as claimed.

{\bf Part II.}  Now we  prove that $f$ is an SGPtime reduction. We start with the following lemma.

\begin{lemma}\label{le:measure_decrease}
For every $m\in\MN$ and $x\in\{0,1\}^\ast$ the following inequality holds:
\begin{equation}\label{eq:factor_rh_reduction}
    \nu_{|f(x)|}(f(x)) \ge \frac{1}{16 |x|^2 g(|x|)} \cdot \mu_{|x|}(x).
\end{equation}
\end{lemma}

\begin{proof}
For every $x\in\{0,1\}^\ast$ there are two possibilities.
If $\mu_{|x|}(x) \le 2^{-|x|}$, then $f(x)=1^{g(|x|)-|0\ovn00x|}0\ovn00x$ and its measure is:
\begin{align*}
\nu_{|f(x)|}(f(x)) & = \frac{1}{g(|x|) 2^{|x|+2\lceil\log_2 |x| \rceil +2}}  \\
 &\ge \frac{1}{g(|x|) 2^{|x|+2\log_2 |x|+3}} = \frac{1}{8|x|^2 g(|x|) 2^{|x|}} \\
&\ge \frac{1}{8 |x|^2 g(|x|)} \cdot \mu_{|x|}(x) \ge \frac{1}{8 |x|^2 g(|x|)} \cdot \mu_{|x|}(x).
\end{align*}
If $\mu_{|x|}(x) > 2^{-|x|}$, then $f(x) = 1^{g(|x|)-|0\ovn01x'|}0\ovn01x'$ and its measure is:
\begin{align*}
\nu_{|f(x)|}(f(x)) &= \frac{1}{g(|x|) 2^{|x'|+2\lceil\log_2 |x'| \rceil +2}} \\
&\ge \frac{1}{g(|x|) 2^{|x'|+2\log_2 |x'|+3}} = \frac{1}{8 |x'|^2 g(|x|) 2^{|x'|}} \\
&\ge \frac{1}{16 |x|^2 g(|x|)} \cdot \mu_{|x|}(x)
\end{align*}
since $|x'|\le |x|$ and $\mu_{|x|}(x) < 2\cdot 2^{-|x'|}$. In each case the inequality (\ref{eq:factor_rh_reduction}) holds.
\end{proof}

By construction of $g$, all elements of $\{0,1\}^n$ are mapped to
elements of size $g(n)$, hence, $f$ is size-invariant. It
follows that a function $f$ is a composition of a CS-reduction
with the polynomial size-growth function $\CS_f(n) = g(n)$ and a
CM-reduction with a polynomial density function
$\frac{1}{16|x|^2g(|x|)}$. Thus, $f$ is an SGPtime reduction.
\end{proof}

Let  $U$ be a universal NTM such that:
\begin{itemize}
\item[(a)]
$U$ accepts inputs of the form $\overline{M}0w$, where
$\overline{M}$ is the  encoding of an NTM $M$ over a binary alphabet and $w \in \Sigma^\ast$;
\item[(b)]
$U$ simulates $M$ on $w$, i.e., $M$ halts on $w$ if and only $U$
halts on $\overline{M}0w$, in which case they both have   the same  answer (the same final configurations);
\item[(c)]
$U$ has a polynomial-time slow-down, i.e., there exists a
polynomial function $s(k)$ such that $T_{M}(w) \ge
T_{U}(\overline{M}0w) / s(|w|)$.
\end{itemize}
See, for example,  \cite{NW} on how such a deterministic  Turing
machine $U$ can be constructed, a nondeterministic one  can be
constructed in a similar way.

\begin{theorem}\label{th:red2BHU}
For every NTM $M$ over a binary alphabet $\Sigma$, there exists a Ptime computable function $h:\MN \to \MN$ and an SGPtime reduction
$DBH(M) \stackrel{f}{\to} DBH(U,h)$.
\end{theorem}

\begin{proof}
Let $M$ be an NTM over $\Sigma$  and $g$ a polynomial longevity guard for $M$.  Define (in the notation above) a function $f: \Sigma^\ast \to \Sigma^\ast$ by
$$f(x) = 1^{g(|x|)s(|x|)-|0 \ovn 0
    \overline{M} 0 x''|} 0 \ovn 0 \overline{M} 0 x'',$$
were $s$ is the polynomial from the description of the machine $U$ above. Put $h(n) = g(n)s(n)$. Clearly, $f$ gives a Ptime reduction $BH(M) \stackrel{f}{\to} BH(U,h)$. To show  that $f$ is an SGPtime reduction, one can argue as in the proof of Theorem \ref{th:red2BH}. To carry over the argument, one needs the following inequality for every $m\in\MN$ and $x\in\{0,1\}^\ast$:
$$    \nu_{|f(x)|}(f(x)) \ge \frac{1}{16 |x|^2 g(|x|)s(|x|) 2^{|\overline{M}|+1}} \cdot \mu_{|x|}(x)$$
which differs from the inequality (\ref{eq:factor_rh_reduction}) by a polynomial factor $2^{|\overline{M}|+1}s(|x|)$ in the denominator. The proof of this is similar to the one in Lemma \ref{le:measure_decrease} and we omit it.
\end{proof}

\begin{corollary}
There exists an NTM $U$ such that $DBH(U)$ is SGPtime complete.
\end{corollary}

\begin{proof}

By Theorem \ref{th:red2BH}   for any $\DistNP$ problem $(D,\mu)$ there
exists an NTM $M$ over a binary alphabet $\Sigma$, a polynomial longevity guard $g$ of $M$,  and an SGPtime reduction of $(D,\mu)$ to $DBH(M,g)$. By Proposition \ref{pr:restriction2BHMg},
there is an SGPtime reduction of $DBH(M,g)$ to $DBH(M)$. By Theorem \ref{th:red2BHU},
there exists a Ptime computable function $h$ and an SGPtime reduction
$DBH(M) \stackrel{f}{\to} DBH(U,h)$. Now, again by Proposition \ref{pr:restriction2BHMg}, there is an SGPtime reduction of $DBH(U,h)$ to $DBH(U)$. Hence, $(D,\mu)$ is SGPtime reducible to $DBH(U)$, as claimed.
\end{proof}


\section{Open problems}
\label{se:Problems}
In this section we discuss some open problems on generic complexity.

\begin{problem}
Is it true that every $\NP$-complete problem is generically in $\P$?
\end{problem}

In fact, even a much stronger version of the question above is still open:
\begin{problem}
Is it true that every $\NP$-complete problem is strongly generically in $\P$?
\end{problem}
Some of the well-known $\NP$-complete problems are in $\GP$, or in $\SGP$, see \cite{GMMU} for examples. However, there is no general approach to this problem at present.
If the answer to one of the questions above (in particular, the second one) is affirmative,
then it will imply that for all practical reasons $\NP$-complete problems are rather easy.
Otherwise, we will have an interesting partition of $\NP$-complete problems into several classes
with respect to their generic behavior.

It was shown in \cite{HM} that the halting problem for  one-end tape Turing machines is  in $\GP$. It remains to be seen if a similar result holds for Turing machines where the tape is infinite at both ends.

\begin{problem}
Is it true that the halting problem is in $\GP$ for Turing machines with one tape
that is infinite at both ends?
\end{problem}

It is known (see \cite{GMMU}) that the classes of functions
that are polynomial on average and generically polynomial are incompatible, i.e., none  of them is a subclass of the other.  Nonetheless, the relationship between $\SGP$-complete and $\NP$-complete on average is still unclear. To this end, the following problem is of interest.

\begin{problem}
Is it true that every $\NP$-complete on average problem is $\SGP$-complete?
\end{problem}


\begin{thebibliography}{10}

\bibitem{Arzhantseva:1998}
G.~{Arzhantseva}.
\newblock {Generic properties of finitely presented groups and Howson's
  theorem}.
\newblock {\em Comm. Algebra}, 26:3783--3792, 1998.

\bibitem{Bienvenu_Day_Holzl:2013}
L.~{Bienvenu}, {Day} A., and R.~{Holzl}.
\newblock {From bi-immunity to absolute undecidability}.
\newblock {\em J. Symbolic Logic}, 78:1218--1228, 2013.

\bibitem{BMR}
A.~V. {Borovik}, A.~G. {Myasnikov}, and V.~N. {Remeslennikov}.
\newblock Multiplicative measures on free groups.
\newblock {\em Int. J. Algebra Comput.}, 13:705--731, 2003.

\bibitem{Champetier:1995}
C.~{Champetier}.
\newblock {Propri\'et\'e statistiques des groupes de pr\'esentation finie}.
\newblock {\em Adv. in Math.}, 116:197--262, 1995.

\bibitem{Downey-Jockusch-McNicholl-Schupp:2014}
R.~{Downey}, C.~{Jockusch}, T.~{McNicholl}, and P.~{Schupp}.
\newblock {Asymptotic density and the Ershov hierarchy}.
\newblock To appear. Available at \url{http://arxiv.org/abs/1309.0137}, 2014.

\bibitem{Downey-Jockusch-Schupp:2013}
R.~{Downey}, C.~{Jockusch}, and P.~{Schupp}.
\newblock {Asymptotic density and computably enumerable sets}.
\newblock {\em J. Math. Log}, 13:43, 2013.

\bibitem{GMMU}
R.~{Gilman}, A.~G. {Myasnikov}, A.~D. {Miasnikov}, and A.~{Ushakov}.
\newblock {Generic complexity of algorithmic problems}.
\newblock In preparation.

\bibitem{GMMU:2007}
R.~{Gilman}, A.~G. {Myasnikov}, A.~D. {Miasnikov}, and A.~{Ushakov}.
\newblock {Report on generic case complexity}.
\newblock Preprint, available at \url{http://arxiv.org/abs/0707.1364}.

\bibitem{Gromov_hyperbolic}
M.~{Gromov}.
\newblock {Hyperbolic groups}.
\newblock In {\em Essays in group theory}, volume~8 of {\em MSRI Publications},
  pages 75--263. Springer, 1985.

\bibitem{Gr}
M.~{Gromov}.
\newblock {Asymptotic invariants of infinite groups}.
\newblock In {\em Geometric Group Theory II}, volume 182 of {\em LMS lecture
  notes}, pages 290--317. Cambridge Univ. Press, 1993.

\bibitem{Gromov:2003}
M.~{Gromov}.
\newblock {Random walks in random groups}.
\newblock {\em Geom. Funct. Analysis}, 13:73--146, 2003.

\bibitem{Gurevich:1991}
Y.~{Gurevich}.
\newblock {Average case completeness}.
\newblock {\em J. Comput. Syst. Sci.}, 42:346--398, 1991.

\bibitem{HM}
J.~D. {Hamkins} and A.~G. {Myasnikov}.
\newblock {The halting problem is decidable on a set of asymptotic probability
  one}.
\newblock {\em Notre Dame Journal of Formal Logic}, 47:515--524, 2006.

\bibitem{Igusa:2013}
G.~{Igusa}.
\newblock {Nonexistence of minimal pairs for generic computability}.
\newblock {\em J. Symbolic Logic}, 78(2):511--522, 2013.

\bibitem{Impagliazzio95}
R.~{Impagliazzo}.
\newblock {A personal view of average-case complexity}.
\newblock In {\em Proceedings of the 10th Annual Structure in Complexity Theory
  Conference (SCT'95)}, pages 134--147, 1995.

\bibitem{Jockusch-Schupp:2012}
C.~{Jockusch} and P.~{Schupp}.
\newblock {Generic computability, Turing degrees, and asymptotic density}.
\newblock {\em J. Lond. Math. Soc.}, 85(2):472--490, 2012.

\bibitem{KMSS1}
I.~{Kapovich}, A.~G. {Miasnikov}, P.~{Schupp}, and V.~{Shpilrain}.
\newblock {Generic-case complexity, decision problems in group theory and
  random walks}.
\newblock {\em J. Algebra}, 264:665--694, 2003.

\bibitem{KMSS2}
I.~{Kapovich}, A.~{Myasnikov}, P.~{Schupp}, and V.~{Shpilrain}.
\newblock Average-case complexity and decision problems in group theory.
\newblock {\em Adv. Math.}, 190:343--359, 2005.

\bibitem{Levin:1986}
L.~{Levin}.
\newblock {Average case complete problems}.
\newblock {\em SIAM J. Comput.}, 15:285--286, 1986.

\bibitem{MRyb}
A.~G. {Miasnikov} and A.~{Rybalov}.
\newblock {On generically undecidable problems}.
\newblock in preparation.

\bibitem{MSU_book:2011}
A.~G. {Miasnikov}, V.~{Shpilrain}, and A.~{Ushakov}.
\newblock {\em Non-Commutative Cryptography and Complexity of Group-Theoretic
  Problems}.
\newblock Mathematical Surveys and Monographs. AMS, 2011.

\bibitem{NW}
T.~{Neary} and D.~{Woods}.
\newblock {Small fast universal Turing machines}.
\newblock technical report NUIM-CS-TR-200511, National university of Ireland,
  Maynooth, 2005.

\bibitem{Ol}
A.~Yu. {Ol'shanskii}.
\newblock Almost every group is hyperbolic.
\newblock {\em Int. J. Alg. Comput.}, 2:1--17, 1992.

\bibitem{papa}
C.~{Papadimitriou}.
\newblock {\em Computational Complexity}.
\newblock Addison-Wesley, 1994.

\bibitem{Wang:1995}
J.~Wang.
\newblock Average-case completeness of a word problem for groups.
\newblock In {\em Proceedings of the twenty-seventh annual ACM symposium on
  Theory of computing}, STOC '95, pages 325--334. ACM, 1995.

\end{thebibliography}
\end{document}